\documentclass[12pt]{iopart}
\usepackage{iopams}

\expandafter\let\csname equation*\endcsname\relax

\expandafter\let\csname endequation*\endcsname\relax

\usepackage{amsmath}

\usepackage{lineno, hyperref}
\modulolinenumbers[5]
\usepackage{amsthm}
\usepackage{amssymb}
\usepackage{stmaryrd}

\newtheorem{theorem}{Theorem}[section]

\newtheorem{proposition}[theorem]{Proposition}

\DeclareMathAlphabet{\mathantt}{OT1}{antt}{li}{it}
\DeclareMathAlphabet{\mathpzc}{OT1}{pzc}{m}{it}

\begin{document}

\title[Integrable couplings of a D-KN hierachy]{Integrable couplings of a generalized D-Kaup-Newell hierarchy and their Hamiltonian structures}

 \author{Morgan McAnally}
\address{Department of Mathematics, The University of Tampa, Tampa, FL 33606-1490 USA}
 \ead{mmcanally@ut.edu} 
\author{Wen-Xiu Ma} 
\address{Department of Mathematics and Statistics, University of South Florida, Tampa, FL 33620-5700, USA} 
\address{College of Mathematics and Systems Science, Shandong University of Science and Technology, Qingdao 266590, Shandong, China}
\address{International Institute for Symmetry Analysis and Mathematical Modeling, Department of Mathematical Sciences, North-West University, Mafikeng Campus, Private Bag X2046, Mmabatho 2735, South Africa}
\address{Department of Mathematics, Zhejiang Normal University, Jinhua 321004, Zhejiang, China,
College of Mathematics and Systems Science, Shandong University of Science and Technology, Qingdao 266590, Shandong, PR China}
  \ead{mawx@cas.usf.edu}
\vspace{10pt}
\begin{indented}
\item[]June 2019
\end{indented}

\begin{abstract}
We enlarge the spectral problem of a generalized D-Kaup-Newell (D-KN) spectral problem.  Solving the enlarged zero-curvature equations, we produce integrable couplings. A reduction of the spectral matrix leads to a second integrable coupling system. Next, bilinear forms that are symmetric, ad-invariant, and non-degenerate on the given non-semisimple matrix Lie algebra are computed to employ the variational identity. The variational identity is then applied to the original enlarged spectral problem of a generalized D-KN hierarchy and the reduced problem. Hamiltonian structures are presented, as well as a bi-Hamiltonian formulation of the reduced problem. Both hierarchies have infinitely many commuting symmetries and conserved densities, i.e., are Liouville integrable.
\end{abstract}

%
% Uncomment for keywords
\vspace{2pc}
%
% Uncomment for Submitted to journal title message
%\submitto{\JPD}
%
% Uncomment if a separate title page is required
\date{\today}

\maketitle

% For two-column output uncomment the next line and choose [10pt] rather than [12pt] in the \documentclass declaration
%\ioptwocol
%

 \section{Introduction}

The finding of new integrable couplings has become an important area of research in mathematical physics \cite{WF}-\cite{ic23}. Originally, integrable couplings were found in the study of centerless Virasoro symmetry algebras of integrable systems \cite{WF,ic2}. Given an integrable system $u_t=K(u)$, an integrable coupling of the system is a triangular system of the form:
\begin{equation}
\label{ts} \bar{u}_t=[u_t, v_t]^T=\bar{K}(\bar{u})=[
K(u), 
S(u,v)]^T
\end{equation}
where potentials $u$ and $v$ are scalar functions or vector functions with dependent variables $\bar{x}=(t,x_1,x_2,...), \bar{u}=[u,v]^T$. It is important that the new differential equations in the bigger system (\ref{ts}) involve the dependent variables $u$ and all of its derivatives, i.e., $\frac{\partial S}{\partial u} \neq 0$. Integrable couplings were first constructed through perturbations \cite{WF,ic2,ic3}, then the spectral matrices were enlarged \cite{ic4,ic5,ic6}, and in 2006, the connection between integrable couplings and semi-direct sums of Lie algebras was realized \cite{ic10,ic11}. Very recently, a novel kind of AKNS integrable couplings was analyzed  \cite{ic23} and other new integrable couplings have been presented  \cite{ic23a, ic23b}. This paper has the enlarge spectral matrix with all submatrices depending on $\lambda$. This technique has only be seen recently. 

A generalized D-KN hierarchy is derived from the following isospectral problem:
\begin{equation}
\label{newhier}
\phi_x=U(\lambda,u) \phi = \begin{bmatrix} \lambda^2 - r_1 & \lambda p_1 + s_1 \\
\lambda q_1 + v_1 & -\lambda^2 + r_1  \end{bmatrix} \phi,     
\end{equation}
where $ u=[p_1 , q_1 , r_1 , s_1 , v_1]^T$ are potentials,
$ \phi=[\phi_1 , \phi_2 ]^T$ and $U \in sl(2, \mathbb{R})$. A hierarchy starting from this spectral problem was analyzed and found to be integrable in the Liouville sense \cite{morgan}. A reduction to the spectral problem \ref{newhier} is bi-Hamiltonian \cite{morgan}. In this paper, we enlarge both the spectral problem \ref{newhier} and its reduction and show their integrability which is a new finding.

 Recall, the D-KN spectral problem is known \cite{7,28} to be
\begin{equation}
\label{dknsp}
\phi_x=U(\lambda, u) \phi = \begin{bmatrix} \lambda^2 +r_1 & \lambda p_1 \\
\lambda q_1 & -\lambda^2- r_1  \end{bmatrix} \phi,  
\end{equation}
which depends on three potentials  $u=[p_1 , q_1,  r_1]^T$ and $\phi$ is the same as (\ref{newhier}) with $U \in sl(2, \mathbb{R})$. The new spectral problem (\ref{newhier}) is a generalization of the D-KN spectral problem  adding two new potentials $s_1$ and $v_1$. Previously, the cases where $r_1=\alpha$ and $r_1=\alpha pq$ have been shown to generate integrable hierarchies \cite{ic26,27} for the D-KN spectral problem (\ref{dknsp}).  It is clear that (\ref{newhier}) is a generalization of the Kaup-Newell \cite{32} spectral matrix, as well. We will note that AKNS hierarchy \cite{33} may be found from (\ref{newhier}) by letting $p_1=q_1=r_1=0$ and choosing a suitable Laurent expansion.  

Two sections complete this paper: integrable couplings and Hamiltonian structures. In the integrable couplings section, an enlarged the spectral matrix (\ref{newhier})is presented. Solving the zero-curvature equation, we find a system of recursive relation and prove they are local. We use this to present the hierarchy of integrable couplings. Next, a reduction of the enlarged spectral matrix leads to a second integrable coupling system. The section of Hamiltonian structures follows where a non-degenerate, ad-invariant, symmetric bilinear forms are found. The bilinear forms along with the variational identity produce Hamiltonian structures of a generalized D-KN integrable couplings and its reduced integrable coupling system. A bi-Hamiltonian structure is found for the reduced integrable couplings. We discover infinitely many commuting symmetries and conserved functionals for both hierarchies and, thus, their Liouville integrability.

\section{Integrable couplings}

In order to simplify notation, we define a triangular block matrix as follows:
\begin{equation}
\label{em}
M(A_1,A_2)=\begin{bmatrix} A_1 & A_2 \\
0 & A_1 \end{bmatrix}.
\end{equation}
It can easily be shown that matrices of this form are closed under matrix multiplication, i.e., constitute a Lie algebra. The associated matrix loop algebra $\tilde{\mathfrak{g}}(\lambda)$ is formed by all block matrices of the type:
\begin{equation}
\label{nla}
\tilde{\mathfrak{g}}(\lambda)=  \{ M(A_1, A_2) |   
 \enspace M \enspace \mbox{defined by} \enspace (\ref{em}), \mbox{entries of $A_i$ are Laurent series in  $\lambda$} \}.
\end{equation}
We will use this notation throughout this paper.
\subsection{Generalized D-KN integrable couplings}
A spectral matrix is chosen from $\tilde{\mathfrak{g}}(\lambda)$ as
\begin{equation}
\label{smic1}
\bar{U}=\bar{U}(\bar{u},\lambda)=M(U(\lambda,u),U_1(\lambda,v)),
\end{equation}
where $\bar{u}=[u,v]^T,
u=[p_1,q_1,r_1,s_1,v_1]^T, v=[p_2,q_2,r_2,s_2,v_2]^T$ are potentials. $U$ is from the generalized D-KN spectral problem (\ref{newhier}) and
\begin{equation}
 U_1=U(\lambda,v)=\begin{bmatrix} \lambda^2 - r_2 & \lambda p_2 + s_2  \\ \lambda q_2 + v_2 & -\lambda^2 + r_2 \end{bmatrix}.
 \end{equation}
The isospectral problem is
\begin{equation}
\label{spic10}
\bar{\phi}_x=\bar{U} \bar{\phi},  
\end{equation}
where $\bar{\phi}= [\psi , \phi ]^T, \psi=[\psi_1, \psi_2]^T$ and $\phi=[\phi_1, \phi_2]^T$. Note that $U$ is the same matrix as (\ref{newhier}).

Assume that the solution to the stationary zero-curvature equation, $\bar{W}_x=[\bar{U},\bar{W}]$, is of the form
\begin{equation}
\label{W}
\bar{W}=\begin{bmatrix} W & W_1 \\ 0 & W \end{bmatrix}\in \tilde{\mathfrak{g}}(\lambda),
W=
\begin{bmatrix} a & b \\  c & -a \end{bmatrix}, W_1=\begin{bmatrix}  e & f \\ g & -e \end{bmatrix}
\end{equation}
then  we get the following matrix formulas:
\begin{equation}
\begin{cases}
W_x=UW-WU, \\
W_{1,x}=U_1W-WU_1+UW_1-W_1U.
\end{cases}
\end{equation}
Solving these two formulas, we get the differential equations:
\begin{equation}
\begin{cases}
\label{rr38}
a_x=-q_1b\lambda+p_1c\lambda-v_1b+s_1c, \\
b_x=-2p_1a\lambda+2b\lambda^2-2s_1a-2r_1b, \\
c_x=2q_1a\lambda-2c\lambda^2+2v_1a+2r_1c,\\
e_x= p_1g\lambda+p_2c\lambda-q_2b\lambda-q_1f\lambda+s_1g+s_2c-v_1f-v_2b, \\
f_x=2b\lambda^2+2f\lambda^2-2p_1e\lambda-2p_2a\lambda-2r_1f-2r_2b-2s_1e-2s_2a, \\
g_x=-2c\lambda^2-2g\lambda^2+2q_1e\lambda+2q_2a\lambda+2r_1g+2r_2c+2v_1e+2v_2a. 
\end{cases}
\end{equation}
By assuming $a,b,c,e,f,g,$ have the following Laurent series expansions:
\begin{equation}
\label{le}
\begin{matrix} a = \sum\limits_{i=0}^{\infty} a_i \lambda^{-i}, & & & b = \sum\limits_{i=0}^{\infty} b_i \lambda^{-i}, & & & c = \sum\limits_{i=0}^{\infty} c_i \lambda^{-i}, \\
e = \sum\limits_{i=0}^{\infty} e_i \lambda^{-i}, & & & f = \sum\limits_{i=0}^{\infty} f_i \lambda^{-i}, & & & g = \sum\limits_{i=0}^{\infty} g_i \lambda^{-i}, \\
 \end{matrix}
\end{equation}
and substituting (\ref{le}) into (\ref{rr38}), we have the recursion relations
\begin{equation}
\label{rr4}
\begin{cases}
b_{i+1}&=\frac{b_{i-1,x}}{2} + p_1a_i+s_1a_{i-1}+r_1b_{i-1},\\
c_{i+1}&=-\frac{c_{i-1,x}}{2} + q_1a_i+v_1a_{i-1}+r_1c_{i-1}, \\
a_{i+1,x}&=-q_1\frac{b_{i,x}}{2}-p_1\frac{c_{i,x}}{2}+(p_1v_1-q_1s_1)a_i-q_1r_1b_i
+p_1r_1c_i+s_1c_{i+1}\\
&-v_1b_{i+1}, \ \\
f_{i+1}&=\frac{f_{i-1,x}}{2}-b_{i+1}+p_2a_i+p_1e_i+s_2a_{i-1}+s_1e_{i-1}+r_2b_{i-1}+r_1f_{i-1}, \\
g_{i+1}&=-\frac{g_{i-1,x}}{2}-c_{i+1}+q_2a_i+q_1e_i+v_2a_{i-1}+v_1e_{i-1}+r_2c_{i-1}+r_1g_{i-1},\\
e_{i+1,x}&=-\frac{g_{i,x}}{2}p_1-\frac{f_{i,x}}{2}q_1+(p_2-p_1)[-\frac{c_{m,x}}{2}+v_1a_m+r_1c_m]\\
&+(q_1-q_2)[\frac{b_{m,x}}{2}+s_1a_m+r_1b_m]+s_1g_{i+1}+s_2c_{i+1}-v_1f_{i+1}-v_2b_{i+1}\\
&+(p_1v_2-q_1s_2)a_i+(p_1v_1-q_1s_1)e_i+p_1r_2c_i-q_1r_2b_i\\
&+p_1r_1g_i-q_1r_1f_i, 
\end{cases}
\end{equation}
for all $i \geq 1$ with initial values
\begin{equation}
\label{iv}
\begin{matrix} a_0=\alpha, &    b_0=c_0=0,&  a_1= 0, & b_1=\alpha p_1, &   c_1=\alpha q_1,\\
                      e_0=\beta, &   f_0=g_0=0,& e_1= 0, & f_1=(\beta-\alpha) p_1+p_2\alpha, & g_1=(\beta-\alpha) q_1+q_2 \alpha, \end{matrix}
\end{equation}
and the conditions for integration
\begin{equation}
\label{conditions}
\begin{cases}
a_i |_{u=0}=b_i |_{u=0}=c_i |_{u=0}=0, \quad i \geq 1, \\
e_i |_{u=0}=f_i |_{u=0}=g_i |_{u=0}=0, \quad i \geq 1,
\end{cases}
\end{equation}
which determine the sequence of $\{a_i,b_i,c_i,e_i,f_i,g_i| i \geq 0\}$ uniquely. For $i=2,3,$ we have the following: 
\begin{flalign*}
b_2& =\alpha s_1, \quad \quad \quad c_2 = \alpha v_1,\quad \quad  a_2 = -\alpha \frac{1}{2}p_1q_1, \\
f_2& =(\beta-\alpha) s_1+\alpha s_2,  \quad \quad  e_2 = (\alpha q_1+\frac{1}{2}\beta q_1 -\frac{1}{2}q_2)p_1-\frac{1}{2}\alpha p_2q_1, \\
g_2& = (\beta- \alpha) v_1+\alpha v_2;  \\
b_3& = \alpha \frac{1}{2} (-p_1^2q_1+2p_1r_1+p_{1,x}), \quad \quad  c_3= -\alpha \frac{1}{2}
 (q_1^2p_1-2q_1r_1+q_{1,x}),\\
a_3&= -\alpha \frac{1}{2} (p_1v_1+q_1s_1),  \\
f_3&= \frac{1}{2}[(\beta-2\alpha)p_{1,x}+\alpha p_{2,x}+((\beta+3\alpha)p_1^2q_1+(2\beta -4\alpha)r_1p_1 \\
&-\alpha(p_2q_1p_1-2r_2p_1+q_2p_1^2)+2\alpha p_2r_1],\\
g_3&=- \frac{1}{2}[(\beta-2\alpha)q_{1,x}+\alpha q_{2,x}- (-\beta+3\alpha)q_1^2p_1\alpha-(2\beta -4\alpha)r_1q_1\\
&+\alpha(q_2p_1q_1-2r_2q_1+p_2q_1^2)-2\alpha q_2r_1],\\
e_3&=(v_1\alpha-\frac{1}{2}v_1\beta-\frac{1}{2}v_2\alpha)p_1+(s_1\alpha-\frac{1}{2}s_1\beta-\frac{1}{2}s_2\alpha)q_1-\frac{1}{2}p_2v_1\alpha-\frac{1}{2}q_2s_1\alpha.\\
%b_4& = \alpha \frac{1}{2} (-p_1^2v_1-2p_1q_1s_1+2r_1s_1+s_{1,x}), \quad  c_4= -\alpha \frac{1}{2}
 %(q_1^2s_1+2q_1p_1v_1-2r_1v_1+v_{1,x}),\\
%a_4&=-\alpha \frac{1}{8}  (-3q_1^2p_1^2+8q_1p_1r_1+4v_1s_1+2q_1p_{1,x}-2p_1q_{1,x}),  \\
%f_4&= \frac{1}{2}[(\beta-2\alpha)s_{1,x}+\alpha s_{2,x}-\alpha p_1^2v_2+(-\beta +3\alpha)v_1p_1^2 \\ %&+(((-2\beta+6\alpha)q_1-2\alpha q_2)s_1-2\alpha(q_1s_2+p_2v_1))p_1+\\
%& (-2\alpha p_2q_1+(-4\alpha+2\beta)r_1+2r_2 \alpha)s_1+2\alpha r_1s_2],\\
%g_4&=-\frac{1}{2}[(\beta-2\alpha)v_{1,x}+\alpha v_{2,x}+\alpha q_1^2s_2-(-\beta +3\alpha)s_1q_1^2 \\ %&-(((-2\beta+6\alpha)p_1-2\alpha p_2)v_1-2\alpha(q_2s_1+p_1v_2))q_1- (-2\alpha q_2p_1+\\
%&(-4\alpha+2\beta)r_1+2r_2 \alpha)v_1-2\alpha r_1v_2],\\
%e_4&=\frac{1}{8} [(-2\beta + 6\alpha)q_1 -2\alpha q_2 )p_{1,x}+((2\beta - 6\alpha)p_1 +2\alpha p_2)q_{1,x} +2\alpha(-
%q_1p_{2,x}+p_1q_{2,x})\\
%&+3q_1((\beta - 4\alpha)q_1+2 \alpha q_2)p_1^2 +(6 \alpha p_2q_1^2+((-8\beta+24\alpha)r_1-8\alpha r_2)q_1\\
%& - 8\alpha q_2r_1)p_1 - 8\alpha q_1p_2r_1+((-4\beta +8 \alpha)v_1  - 4\alpha v_2)s_1 -4\alpha s_2v_1]. 
\end{flalign*}
All $\{a_i,b_i,c_i,e_i,f_i,g_i| i \geq 0\}$ can be proven as differential polynomials of $\bar{u}$ with respect to $x$. 
\begin{proposition}
\label{proofs}
Let $ \{ a_i,b_i,c_i,e_i,f_i,g_i| i=0,1 \}$ be given by equations (\ref{iv}). Then all functions $\{a_i,b_i,c_i,e_i,f_i,g_i| i \geq 0 \}$ determined by equation (\ref{rr4}) with the conditions (\ref{conditions}) are differential polynomials in $\bar{u}$ with respect to $x$, and thus, are local.
\end{proposition}
\begin{proof}
We compute from the enlarged stationary zero-curvature equation, $\bar{W}_x=[\bar{U},\bar{W}]$, 
\begin{equation}
\frac{d}{dx} \text{tr}(\bar{W}^2)=2 \text{tr}(\bar{W}\bar{W}_x)=2\text{tr}(\bar{W}[\bar{U},\bar{W}])=2(\text{tr}(\bar{W}^2\bar{U})-\text{tr}(\bar{W}^2\bar{U}))=0,
\end{equation}
and seeing that the $\text{tr}(\bar{W}^2)=4(a^2+bc)$, we have
\begin{equation}
a^2+bc=(a^2+bc)|_{u=0}=\alpha^2,
\end{equation}
following from the initial data (\ref{iv}). Now, we use (\ref{le}), the Laurent expansions of $a,b,c$, to give
\begin{equation}
\label{proof}
a_i=\frac{\alpha}{2}- \frac{1}{2\alpha} \sum_{k+l=i, k,l \geq 1} a_ka_l -\frac{1}{2\alpha} \sum_{k+l=i, k,l \geq 0} b_kc_l, i \geq 1.
\end{equation}
Based on the recursion relation above (\ref{proof}) and the previous (\ref{rr4}), we use mathematical induction to see that all functions $\{a_i,b_i,c_i, i \geq 0\}$ are differential polynomials in $u$ with respect to $x$, and therefore, are local.

Now, we have
\begin{align*}
\frac{d}{dx}(2ae+fc+gb)=& 2a_xe+2ae_x+f_xc+fc_x+g_xb+gb_x \\
  =& 2e(-q_1b\lambda+p_1c\lambda-v_1b+s_1c)+2a( p_1g\lambda+p_2c\lambda\\
& -q_2b\lambda-q_1f\lambda+s_1g+s_2c-v_1f-v_2b)+c(2b\lambda^2\\
&+2f\lambda^2-2p_1e\lambda-p_2a\lambda-2r_1f-2r_2b-2s_1e-2s_2a )\\
& +f(2q_1a\lambda-2c\lambda^2 +2v_1a+2r_1c)+b(-2c\lambda^2-2g\lambda^2 \\
&+2q_1e\lambda+q_2a\lambda+2r_1g+2r_2c+2v_1e+2v_2a)\\
&+ g(-2p_1a\lambda+2b\lambda^2-2s_1a-2r_1b)=0.
\end{align*}
Similarly, we get
$$2ae+fc+gb=(2ae+fc+gb)|_{\bar{u}=0}=\alpha \beta. $$
Therefore, using the Laurent expansions of $a, b, c, e, f,$ and $g$ in (\ref{le}), we have
\begin{equation}
\label{proof1}
e_i=\beta-\frac{\beta}{\alpha}a_i-\frac{1}{2\alpha} \sum_{k+l=i, k,l \geq 0} f_kc_l-\frac{1}{2\alpha} \sum_{k+l=i, k,l \geq 0} g_kb_l-\frac{1}{\alpha}\sum_{k+l=i, k,l \geq 1} a_ke_l,
\end{equation}
for all $i \geq 1$. Using the localness of $\{a_i,b_i,c_i| i \geq 0 \}$ and the recursive relations (\ref{rr4}) and (\ref{proof1}), we may see through mathematical induction that all functions $\{e_i,f_i,g_i| i \geq 0 \}$ are differential polynomials in $\bar{u}$ with respect to $x$.
 This completes the proof.
\end{proof} 

Now, we need to solve the zero-curvature equations,
\begin{equation}
\label{zce}
\bar{U}_{t_m}-\bar{V}_x^{[m]}+[\bar{U},\bar{V}^{[m]}]=0, \quad m \geq 0,
\end{equation}
which are the compatibility conditions between (\ref{spic10}) and the temporal problems, 
\begin{equation}
\label{spic11}
\bar{\phi}_{t_m}=\bar{V}^{[m]} \bar{\phi} = \bar{V}^{[m]} (\bar{u},\lambda) \bar{\phi} , \quad m \geq 0.
\end{equation}
In order to do this, we introduce a series of Lax operators
\begin{equation}
\label{lax}
\bar{V}^{[m]}(\bar{u},\lambda)=( \lambda^m \bar{W})_{+}.
\end{equation}
After solving (\ref{zce}), we generate a hierarchy of soliton equations, for all $m \geq 0$,
\begin{equation}
\label{shier1}
\bar{u}_{t_m}=\bar{K}_m=  \begin{bmatrix}
2b_{m+1}\\
-2c_{m+1}\\
q_1b_{m+1}-p_1c_{m+1} \\
-2p_1a_{m+1}+2b_{m+2}\\
2q_1a_{m+1}-2c_{m+2}\\
2f_{m+1}+2b_{m+1} \\
-2g_{m+1}-2c_{m+1}\\
q_1f_{m+1}+q_2b_{m+1}-p_1g_{m+1}- p_2c_{m+1} \\
-2p_1e_{m+1}-2p_2a_{m+1}+2b_{m+2}+2f_{m+2}\\
2q_1e_{m+1}+2q_2a_{m+1}-2c_{m+2}-2g_{m+2} \end{bmatrix}.
\end{equation}
We have
\begin{equation}
\label{recurr}
\bar{K}_m=\bar{\Phi} \bar{K}_{m-1}=\bar{\Phi}^m\bar{K}_0 , \quad m \geq 0,
\end{equation}
where
\begin{equation}
\label{phi1}
\bar{\Phi}=\begin{bmatrix} \Phi & 0 \\
\Phi_1 - \Phi & \Phi \end{bmatrix}.
\end{equation}
where $\Phi$ is the recursion operator of the original system $u_t=K(u)$ equal to
\begin{equation}
\label{phi10a}
\begin{bmatrix} -p_1 \partial^{-1}v_1   &  -p_1  \partial^{-1}s_1   & 2s_1  \partial^{-1} & 1-p_1  \partial^{-1}q_1  
& -p_1 \partial^{-1}p_1 \\
-s_1 \partial^{-1}q_1 & -s_1 \partial^{-1}p_1  & & & \\
 & & & & \\
q_1 \partial^{-1}v_1 & q_1 \partial^{-1}s_1  & -2v_1\partial^{-1} & q_1 \partial^{-1}q_1 & 1+q_1 \partial^{-1}p_1 \\
+ v_1\partial^{-1}q_1 & +v_1\partial^{-1}p_1 & & & \\
 & & & & \\
p_1v_1\partial^{-1}\dfrac{q_1}{2} & p_1v_1\partial^{-1}\dfrac{p_1}{2} & -p_1v_1\partial^{-1} & \dfrac{q_1}{2}  &  \dfrac{p_1}{2} \\
-q_1s_1\partial^{-1}\dfrac{q_1}{2} & -q_1s_1\partial^{-1}\dfrac{p_1}{2}  & +q_1s_1\partial^{-1}& & \\
 & & & & \\
 \frac{1}{2}\partial+r_1-s_1 \partial^{-1}v_1& -s_1 \partial^{-1}s_1  & 2p_1r_1\partial^{-1} & -s_1 \partial^{-1}q_1  
& -s_1 \partial^{-1}p_1 \\
 -p_1r_1\partial^{-1}q_1  &  -\partial p_1 \partial^{-1} \dfrac{p_1}{2} &+\partial p_1 \partial^{-1}  & & \\
 -\partial p_1\partial^{-1}\dfrac{q_1}{2}&-p_1r_1\partial^{-1}p_1 & & & \\
& & & & \\
v_1 \partial^{-1}v_1 & -\frac{1}{2}\partial+r_1+v_1 \partial^{-1}s_1 & -2q_1r_1\partial^{-1} & v_1 \partial^{-1}q_1 & v_1 \partial^{-1}p_1 \\
+\partial q_1 \partial^{-1}\dfrac{q_1}{2} &+ q_1r_1\partial^{-1}p_1  & -\partial q_1 \partial^{-1} & & \\
+q_1r_1 \partial^{-1}q_1 & +\partial q_1 \partial^{-1}\dfrac{p_1}{2}  & & & 
\end{bmatrix}
\end{equation}
%with entries
%\begin{equation}
%\label{phi1}
%\begin{aligned}
%&[\Phi]_{11}= -p_1 \partial^{-1}v_1  -s_1 \partial^{-1}q_1, \;   
%[\Phi]_{12}= -p_1  \partial^{-1}s_1  -s_1 \partial^{-1}p_1, \\
%&[\Phi]_{13}= 2s_1  \partial^{-1}, \;
%[\Phi]_{14}= 1-p_1  \partial^{-1}q_1,  \;
%[\Phi]_{15}= -p_1 \partial^{-1}p_1, \\
%&[\Phi]_{21}= q_1 \partial^{-1}v_1 +  v_1\partial^{-1}q_1, \;
%[\Phi]_{22}= q_1 \partial^{-1}s_1 +v_1\partial^{-1}p_1, \; 
%[\Phi]_{23}= -2v_1\partial^{-1}, \\ 
%&[\Phi]_{24}= q_1 \partial^{-1}q_1, \;
%[\Phi]_{25}= 1+q_1 \partial^{-1}p_1, \\
%&[\Phi]_{31}= (p_1v_1-q_1s_1)\partial^{-1}\dfrac{q_1}{2}, \;
%[\Phi]_{32}=   (p_1v_1-q_1s_1)\partial^{-1}\dfrac{p_1}{2}, \\
%&[\Phi]_{33}= - (p_1v_1-q_1s_1)\partial^{-1}, \;
%[\Phi]_{34}= \dfrac{q_1}{2}, \; 
%[\Phi]_{35}=  \dfrac{p_1}{2}, \\
%&[\Phi]_{41}= \frac{1}{2}\partial+r_1-s_1 \partial^{-1}v_1 -p_1r_1\partial^{-1}q_1   -\partial p_1\partial^{-1}\dfrac{q_1}{2},\\
%&[\Phi]_{42}=-s_1 \partial^{-1}s_1-\partial p_1 \partial^{-1} \dfrac{p_1}{2}-p_1r_1\partial^{-1}p_1, \\
%&[\Phi]_{43}= 2p_1r_1\partial^{-1}+\partial p_1 \partial^{-1}, \;
%[\Phi]_{44}= -s_1 \partial^{-1}q_1, \; 
%[\Phi]_{45}= -s_1 \partial^{-1}p_1 \\ 
%&[\Phi]_{51}= v_1 \partial^{-1}v_1+\partial q_1 \partial^{-1}\dfrac{q_1}{2}+q_1r_1 \partial^{-1}q_1, \\
%&[\Phi]_{52}=  -\frac{1}{2}\partial+r_1+v_1 \partial^{-1}s_1+ q_1r_1\partial^{-1}p_1+\partial q_1 \partial^{-1}\dfrac{p_1}{2}, \\
%&[\Phi]_{53}= -2q_1r_1\partial^{-1}-\partial q_1 \partial^{-1}, \; 
%[\Phi]_{54}= v_1 \partial^{-1}q_1, \;
%[\Phi]_{55}= v_1 \partial^{-1}p_1 \\
%\end{aligned}
%\end{equation}
and $\Phi_1$ is the supplemental matrix differential operator with entries
%\begin{equation}
%\label{phi10}
%\Phi_1=\begin{bmatrix} [\Phi_1]_{11} & [\Phi_1]_{12} & [\Phi_1]_{13} & [\Phi_1]_{14} & [\Phi_1]_{15} \\
%[\Phi_1]_{21} & [\Phi_1]_{22} & [\Phi_1]_{23} & [\Phi_1]_{24} & [\Phi_1]_{25} \\
%[\Phi_1]_{31} & [\Phi_1]_{32} & [\Phi_1]_{33} & [\Phi_1]_{34} & [\Phi_1]_{35} \\
%[\Phi_1]_{41} & [\Phi_1]_{42} & [\Phi_1]_{43} & [\Phi_1]_{44} & [\Phi_1]_{45} \\
%[\Phi_1]_{51} & [\Phi_1]_{52} & [\Phi_1]_{53} & [\Phi_1]_{54} & [\Phi_1]_{55} 
%\end{bmatrix},
%\end{equation}
\begin{equation}
\label{phi10b}
\begin{aligned} 
&[\Phi_1]_{11}=-p_1\partial^{-1} v_2 - p_2 \partial^{-1} v_1 - s_2 \partial^{-1} q_1 -s_1 \partial^{-1}q_2 ,\\
& [\Phi_1]_{12}=  -p_1\partial^{-1} s_2 - p_2 \partial^{-1} s_1 - s_2 \partial^{-1} p_1 -s_1 \partial^{-1}p_2 \\ &[\Phi_1]_{13}= 2 s_2 \partial^{-1}+2s_1 \partial^{-1}, [\Phi_1]_{14}= 1  -p_1\partial^{-1} q_2 - p_2 \partial^{-1} q_1, \\
&[\Phi_1]_{15}= -p_1\partial^{-1} p_2 - p_2 \partial^{-1} p_1, & \\
&[\Phi_1]_{21}=q_1\partial^{-1} v_2 + q_2 \partial^{-1} v_1 +v_1 \partial^{-1} q_2 +v_2\partial^{-1}q_1 ,\\
&[\Phi_1]_{22}=  q_1\partial^{-1} s_2 +q_2 \partial^{-1} s_1+v_1 \partial^{-1} p_2 +v_2\partial^{-1}p_1, \\
&[\Phi_1]_{23}= -2v_1 \partial^{-1} -2 v_2\partial^{-1},\\
&[\Phi_1]_{24}=    q_1\partial^{-1} q_2 + q_2 \partial^{-1} q_1, [\Phi_1]_{25}= 1+q_1\partial^{-1} p_2 + q_2 \partial^{-1} p_1, \\
&[\Phi_1]_{31}= \dfrac{(v_1p_2-q_2s_1)}{2}\partial^{-1}q_1+\dfrac{(p_1v_2-q_1s_2)}{2}\partial^{-1}q_1- \dfrac{(p_1v_1-q_1s_1)}{2}\partial^{-1}q_1\\
&+\dfrac{(p_1v_1-q_1s_1)}{2}\partial^{-1}q_2 \\
&[\Phi_1]_{32}= \dfrac{(v_1p_2-q_2s_1)}{2}\partial^{-1}p_1+\dfrac{(p_1v_2-q_1s_2)}{2}\partial^{-1}p_1- \dfrac{(p_1v_1-q_1s_1)}{2}\partial^{-1}p_1\\
&+\dfrac{(p_1v_1-q_1s_1)}{2}\partial^{-1}p_2 \\
&[\Phi_1]_{33}= -(p_1v_2-q_1s_2)\partial^{-1}-(p_2v_1-q_2s_1)\partial^{-1}\\
&[\Phi_1]_{34}= \frac{q_2}{2}, [\Phi_1]_{35}= \frac{p_2}{2}\\ 
&[\Phi_1]_{41}=r_2-s_1\partial^{-1} v_2 - s_2 \partial^{-1} v_1- \partial \dfrac{(p_2-p_1)}{2} \partial^{-1} q_1 -(r_1p_2+p_1r_2)\partial^{-1}q_1\\
&+p_1r_1\partial^{-1}q_1-r_1p_1\partial^{-1}q_2 -\partial \dfrac{p_1}{2} \partial^{-1}q_2, \\
&[\Phi_1]_{42}=  -s_1\partial^{-1} s_2 - s_2 \partial^{-1} s_1 - \partial \dfrac{(p_2-p_1)}{2} \partial^{-1} p_1-(r_1p_2+p_1r_2)\partial^{-1}p_1 \\
&+p_1r_1\partial^{-1}p_1-r_1q_1\partial^{-1}p_2 -\partial \dfrac{p_1}{2} \partial^{-1}p_2, \\
&[\Phi_1]_{43}= \partial (p_2-p_1) \partial^{-1} +2(r_1p_2+p_1r_2)\partial^{-1}+\partial p_1 \partial^{-1}\\
&[\Phi_1]_{44}= -s_1\partial^{-1} q_2 - s_2 \partial^{-1} q_1, [\Phi_1]_{45}= -s_1\partial^{-1} p_2 - s_2 \partial^{-1} p_1 \\
&[\Phi_1]_{51}=v_1\partial^{-1} v_2 +v_2 \partial^{-1} v_1 + \partial \dfrac{(q_1-q_2)}{2} \partial^{-1} q_1+(r_1q_2+q_1r_2)\partial^{-1}q_1 \\
&-q_1r_1\partial^{-1}q_1+r_1q_1\partial^{-1}q_2 -\partial \dfrac{q_1}{2} \partial^{-1}q_2, \\  
&[\Phi_1]_{52}= r_2+ v_1\partial^{-1} s_2 +v_2 \partial^{-1} s_1 + \partial \dfrac{(q_1-q_2)}{2} \partial^{-1} p_1 +(r_1q_2+q_1r_2)\partial^{-1}p_1\\
&-q_1r_1\partial^{-1}p_1+r_1q_1\partial^{-1}p_2 -\partial \dfrac{q_1}{2} \partial^{-1}p_2 \\   
&[\Phi_1]_{53}=- \partial(q_1-q_2) \partial^{-1}  -2(r_1q_2+q_1r_2)\partial^{-1}+\partial q_1 \partial^{-1}, \\
&[\Phi_1]_{54}=   v_1\partial^{-1} q_2 +v_2 \partial^{-1} q_1, [\Phi_1]_{55}= v_1\partial^{-1} p_2 +v_2 \partial^{-1} p_1 
\end{aligned}
\end{equation}
with $\partial=\dfrac{\partial}{\partial x}$ and $\partial^{-1}$ as the inverse operator of $\partial$. 

\subsection{A specific reduction with two less potentials}
\label{bhr}

A spectral matrix, $\bar{U}$, chosen from $\tilde{\mathfrak{g}}(\lambda)$, is of the form: 
\begin{equation}
\label{smic2}
\bar{U}=M(U(\lambda,u),U_1(\lambda,v))
\end{equation}
where $r_1, r_2$ from (\ref{smic1}) are replaced with $\widetilde{r_1}=\frac{1}{2}p_1q_1, \widetilde{r_2}= \frac{1}{2}(p_1q_2+p_2q_1-p_1q_1)$, and $\bar{u}=(u,v)^T$, $u=[p_1,q_1,s_1,v_1]^T, v=[p_2,q_2,s_2,v_2]^T$ are potentials. The corresponding spacial spectral  problem is
\begin{equation}
\label{spic20}
\bar{\phi}_x=\bar{U}(\bar{u},\lambda) \bar{\phi}, 
\end{equation}
where $ \bar{\phi}= [\psi, \phi ]^T, \psi=[\psi_1, \psi_2]^T$ and $\phi=[\phi_1, \phi_2]^T$.

Again, we assume that the solution to the stationary zero-curvature equation, $\bar{W}_x=[\bar{U},\bar{W}]$, is of the form as (\ref{W}). Solving the stationary zero-curvature equation (\ref{zce}), we have the following differential equations as (\ref{rr38}) with $r_1=\widetilde{r_1}$ and $r_2=\widetilde{r_2}$. By assuming $a,b,c,e,f,g,$ have the Laurent expansions (\ref{le}), we have the recursion relations
\begin{equation}
\label{rr8}
\begin{cases}
b_{i+1}&=\frac{b_{i-1,x}}{2} + p_1a_i+s_1a_{i-1}+\frac{1}{2}p_1q_1b_{i-1},\\
c_{i+1}&=-\frac{c_{i-1,x}}{2} + q_1a_i+v_1a_{i-1}+\frac{1}{2}p_1q_1c_{i-1}, \\
a_{i+1,x}&=-q_1\frac{b_{i,x}}{2}-p_1\frac{c_{i,x}}{2}+(p_1v_1-q_1s_1)a_i-\frac{1}{2}p_1q_1^2b_i\\
&+\frac{1}{2}p_1^2q_1c_i+s_1c_{i+1}-v_1b_{i+1}, \ \\
f_{i+1}&=\frac{f_{i-1,x}}{2}-b_{i+1}+p_2a_i+p_1e_i+s_2a_{i-1}+s_1e_{i-1}\\
&+\frac{1}{2}(p_1q_2+p_2q_1-p_1q_1)b_{i-1}+\frac{1}{2}p_1q_1f_{i-1}, \\
g_{i+1}&=-\frac{g_{i-1,x}}{2}-c_{i+1}+q_2a_i+q_1e_i+v_2a_{i-1}+v_1e_{i-1}\\
&+\frac{1}{2}(p_1q_2+p_2q_1-p_1q_1)c_{i-1}+\frac{1}{2}p_1q_1g_{i-1},\\
e_{i+1,x}&=-\frac{p_1g_{i,x}}{2}-\frac{q_1f_{i,x}}{2}+(p_2-p_1)[-\frac{c_{m,x}}{2}+v_1a_m+\frac{1}{2}p_1q_1c_m]\\
&+(q_1-q_2))[\frac{b_{m,x}}{2}+s_1a_m+\frac{1}{2}p_1q_1b_m]\\
&+s_1g_{i+1}+s_2c_{i+1}-v_1f_{i+1}-v_2b_{i+1}+(p_1v_2-q_1s_2)a_i\\
&+(p_1v_1-q_1s_1)e_i+\frac{1}{2}p_1(p_1q_2+p_2q_1-p_1q_1)c_i \\
&-\frac{1}{2}q_1(p_1q_2+p_2q_1-p_1q_1)b_i+\frac{1}{2}p_1^2q_1g_i-\frac{1}{2}p_1q_1^2f_i, 
\end{cases}
\end{equation}
for all $i \geq 1$ with the same initial values (\ref{iv}) and conditions for integration (\ref{conditions}) which determine the sequence of $\{a_i,b_i,c_i,e_i,f_i,g_i| i \geq 0\}$ uniquely. All $\{a_i,b_i,c_i,e_i,f_i,g_i\}$ can be proven as differential polynomials of $\bar{u}$ with respect to $x$.

\begin{proposition}
\label{proofs2}
Let $\{a_i,b_i,c_i,e_i,f_i,g_i| i=0,1 \}$ be given by equation (\ref{iv}). Then all functions $\{a_i,b_i,c_i,e_i,f_i,g_i| i \geq 0 \}$ determined by equations (\ref{rr8}) with the conditions (\ref{conditions}) are differential polynomials in $\bar{u}$ with respect to $x$, and thus, are local.
\end{proposition}
\begin{proof}
For brevity, we leave the proof out. It is similar to Proposition \ref{proofs}.
\end{proof} 

We  solve the zero-curvature equations (\ref{zce}) with the Lax matrices (\ref{lax}) to generate a hierarchy of soliton equations for all $m \geq 0$,
\begin{equation}
\label{shier2}
\bar{u}_{t_m}=\bar{K}_m=  \begin{bmatrix}
2b_{m+1}\\
-2c_{m+1}\\
-2p_1a_{m+1}+2b_{m+2}\\
2q_1a_{m+1}-2c_{m+2}\\
2f_{m+1}+2b_{m+1} \\
-2g_{m+1}-2c_{m+1}\\
-2p_1e_{m+1}-2p_2a_{m+1}+2b_{m+2}+2f_{m+2}\\
2q_1e_{m+1}+2q_2a_{m+1}-2c_{m+2}-2g_{m+2} \end{bmatrix}.
\end{equation}
We have
\begin{equation}
\bar{K}_m=\bar{\Phi} \bar{K}_{m-1}=\bar{\Phi}^m\bar{K}_0 , \quad m \geq 0,
\end{equation}
where $\bar{\Phi}$ is a recursion operator determined from (\ref{rr8}) and given by
\begin{equation}
\label{ro}
\bar{\Phi}=\begin{bmatrix} \Phi & 0 \\
\Phi_1 - \Phi & \Phi \end{bmatrix}.
\end{equation}
The matrix blocks of $\bar{\Phi}$ are defined by 
\begin{equation}
\label{phi2}
\Phi=\begin{bmatrix} -p_1 \partial^{-1}v_1   &  -p_1  \partial^{-1}s_1    & 1-p_1  \partial^{-1}q_1  & -p_1  \partial^{-1}p_1  \\
q_1  \partial^{-1}v_1 & q_1  \partial^{-1}s_1   & q_1  \partial^{-1}q_1  & 1+q_1  \partial^{-1}p_1  \\
 \frac{1}{2}\partial+\widetilde{r_1} -s_1  \partial^{-1}v_1 & -s_1  \partial^{-1}s_1   & -s_1  \partial^{-1}q_1   & -s_1  \partial^{-1}p_1  \\
v_1  \partial^{-1}v_1  & -\frac{1}{2}\partial+\widetilde{r_1} +v_1  \partial^{-1}s_1 & v_1  \partial^{-1}q_1  & v_1  \partial^{-1}p_1  \end{bmatrix} \\
\end{equation}
and
%\begin{align*} 
%[\Phi_1]_{11}=&-p_1\partial^{-1} v_2 - p_2 \partial^{-1} v_1   &[\Phi_1]_{12}= & -p_1\partial^{-1} s_2 - p_2 \partial^{-1} %s_1  &[\Phi_1]_{13}=& 0\\
%[\Phi_1]_{14}= & 1  -p_1\partial^{-1} q_2 - p_2 \partial^{-1} q_1 &[\Phi_1]_{15}=& -p_1\partial^{-1} p_2 - p_2 %\partial^{-1} p_1 \\
%[\Phi_1]_{21}=&q_1\partial^{-1} v_2 + q_2 \partial^{-1} v_1   &[\Phi_1]_{22}= & q_1\partial^{-1} s_2 +q_2 %\partial^{-1} s_1  &[\Phi_1]_{23}=& 0\\
%[\Phi_1]_{24}= &   p_1\partial^{-1} q_2 + p_2 \partial^{-1} q_1 &[\Phi_1]_{25}=& 1+p_1\partial^{-1} p_2 + p_2 %\partial^{-1} p_1 \\
%[\Phi_1]_{31}=& 0 & [\Phi_1]_{32}=& 0 & [\Phi_1]_{33}=& 0 \\
%[\Phi_1]_{34}=& \frac{q_2}{2} & [\Phi_1]_{35}=& \frac{p_2}{2}\\
%[\Phi_1]_{41}=&\frac{1}{2}(p_1q_2+p_2q_1-p_1q_1)-s_1\partial^{-1} v_2 - s_2 \partial^{-1} v_1   &[\Phi_1]_{42}= & 
%-s_1\partial^{-1} s_2 - s_2 %\partial^{-1} s_1  &[\Phi_1]_{43}=& 0\\
%[\Phi_1]_{44}= &   -s_1\partial^{-1} q_2 - s_2 \partial^{-1} q_1 &[\Phi_1]_{45}=& -s_1\partial^{-1} p_2 - s_2 %\partial^{-1} p_1 \\
%[\Phi_1]_{51}=&v_1\partial^{-1} v_2 +v_2 \partial^{-1} v_1   &[\Phi_1]_{52}= &\frac{1}{2}(p_1q_2+p_2q_1-p_1q_1)+ %v_1\partial^{-1} s_2 +v_2 %\partial^{-1} s_1  &[\Phi_1]_{53}=& 0\\
%[\Phi_1]_{54}= &  v_1\partial^{-1} q_2 +v_2 \partial^{-1} q_1 &[\Phi_1]_{55}=& v_1\partial^{-1} p_2 +v_2 %\partial^{-1} p_1 
%\end{align*}
\begin{equation}
\label{phi20}
\Phi_1=\begin{bmatrix}-p_1\partial^{-1} v_2    &  -p_1\partial^{-1} s_2   & 1  -p_1\partial^{-1} q_2  & -p_1  \partial^{-1}p_2  \\
  -p_2 \partial^{-1} v_1   &  -p_2 \partial^{-1} s_1    &  -p_2 \partial^{-1} q_1  & -p_2  \partial^{-1}p_1  \\
 & &  & \\
q_1  \partial^{-1}v_2 & q_1  \partial^{-1}s_2   & q_1  \partial^{-1}q_2  & 1+q_1  \partial^{-1}p_2  \\
+q_2  \partial^{-1}v_1 & +q_2  \partial^{-1}s_1      & +q_2  \partial^{-1}q_1  & +q_2  \partial^{-1}p_1  \\
 & &  & \\
 \widetilde{r_2} -s_1  \partial^{-1}v_2 & -s_1  \partial^{-1}s_2  & -s_1  \partial^{-1}q_2   & -s_1  \partial^{-1}p_2  \\
  -s_2  \partial^{-1}v_1 & -s_2  \partial^{-1}s_1     & -s_2  \partial^{-1}q_1   & -s_2  \partial^{-1}p_1  \\
 & &  & \\
v_1  \partial^{-1}v_2  & \widetilde{r_2} +v_2  \partial^{-1}s_1  & v_1  \partial^{-1}q_2  & v_1  \partial^{-1}p_2  \\
+v_2  \partial^{-1}v_1  & +v_1  \partial^{-1}s_2 &+v_2  \partial^{-1}q_1  & +v_2  \partial^{-1}p_1  \\ \end{bmatrix} 
\end{equation}
where $\widetilde{r_1}= \frac{1}{2}p_1q_1$ and $\widetilde{r_2}= \frac{1}{2}(p_1q_2+p_2q_1-p_1q_1)$.

A specific example can be found from the reduced hierarchy of integrable couplings (\ref{shier2}) when $m=6$ by setting the eight potentials and $\alpha$ and $\beta$ to be the following: $\{ p_1=q_1=0, s_1=u, v_1=-u, p_2=q_2=v, s_2=w, v_2=r, \alpha=-4, \beta=-8 \}$.
We find a coupled mKdV \cite{ic14,ic15} system of equations:
\begin{equation}
\label{cmkdv}
\begin{cases}
u_t=-u_{xxx}-6u^2u_x, \\
v_t=-v_{xxx}-4uvu_x-2u^2v_x+(4r-4w)u^2, \\
w_t=-w_{xxx}+u_{xxx}+(6u^2+(4v-4w)u)u_x-4u^2v_x-2u^2w_x, \\
r_t=-r_{xxx}-u_{xxx}+(-6u^2+(-4r+4v)u)u_x-2u^2r_x-2u^2v_{x}.
\end{cases}
\end{equation}
%Or, we may find a coupled NLSE \cite{ } by letting $\{ p_1=q_1=0, s_1=u, v_1=-u, p_2=q_2=0, s_2=v, v_2=-v, \alpha=1, %\beta=1 \}$ when $m=4$ in reduced hierarchy (\ref{shier}), i.e., 
%\begin{equation}
%\begin{cases}
%u_t=\frac{1}{2}u_{xx}-u|u|, \\
%v_t=-\frac{1}{2}u_{xx}+\frac{1}{2}v_{xx}+u^2v.
%\end{cases}
%\end{equation}

\section{Hamiltonian structures}
\subsection{Constructing bilinear forms over a non-semisimple Lie algebra}
\label{bf}
For generating Hamiltonian structures for the integrable couplings in (\ref{shier1}) and (\ref{shier2}), we use the variational identity over the non-semisimple Lie algebra $  \tilde{\mathfrak{g}}(\lambda)$ \cite{ic12,ic16,ic17}. The variational identity is a generalization of the trace identity \cite{ic12, 31}. The variational identity may be applied to non-semisimple Lie algebras where the trace identity works for semsimple. The idea is to use non-degenerate, symmetric, and ad-invariant bilinear forms on the Lie algebra. We begin with a  construction of non-degenerate, symmetric, and ad-invariant bilinear forms on $ \tilde{\mathfrak{g}}(\lambda)$ by rewriting $ \tilde{\mathfrak{g}}(\lambda)$ into a vector form. 

The isomorphism
\begin{equation}
\sigma :  \tilde{\mathfrak{g}}(\lambda) \to \mathbb{R}^6, A \mapsto (a_1, ... , a_6)^T ,
\end{equation}
where
\begin{equation}
A=M(A_1,A_2) \in \tilde{\mathfrak{g}}(\lambda), \quad A_i = \left[ \begin{matrix} a_{3i-2} & a_{3i-1}\\
a_{3i} & -a_{3i-2}  \end{matrix}\right], i= 1,2,
\end{equation}
and a constant symmetric matrix,
\begin{equation}
F=F_1 \otimes \begin{bmatrix} 1 & 0\\ 0 & 0 \end{bmatrix} + F_2 \otimes \begin{bmatrix} 0 & 1\\ 1 & 0 \end{bmatrix} ,  F_i=\begin{bmatrix} 2 \eta_i & 0 & 0  \\
0 & 0 & \eta_i  \\
0 & \eta_1i& 0  \end{bmatrix}, i=1,2,
\end{equation}
where $\otimes$ is the Kronecker prodect, with arbitrary constants $\eta_1$ and $\eta_2$ furnish the bilinear forms on $ \tilde{\mathfrak{g}}(\lambda)$ defined as
\begin{equation}
\label{bf}
\begin{split}
\langle A,B \rangle_{\tilde{\mathfrak{g}}(\lambda)} =& \langle \sigma(A), \sigma(B) \rangle_{\mathbb{R}^6}\\
=&(a_1, ..., a_6)F(b_1, ..., b_6)^T\\
=&(2a_1b_1+a_2b_3+a_3b_2)\eta_1 + (2a_1b_4+a_2b_6 + a_3b_5 \\
&+ 2a_4b_1 + a_5b_3 + a_6b_2)\eta_2.
\end{split}
\end{equation}
The bilinear forms (\ref{bf}) are symmetric and ad-invariant due to the isomorphism $\sigma$. The bilinear forms, defined by (\ref{bf}), are non-degenerate iff the determinant of F is not zero, i.e.,
\begin{equation}
det(F)=-4 \eta_2^6 \neq 0.
\end{equation}
Therefore, we choose $\eta_2 \neq 0$ to obtain the required non-degenerate, symmetric, and ad-invariant bilinear forms over the enlarged matrix loop algebra $\tilde{\mathfrak{g}}(\lambda)$. For simplicity, we choose $\eta_1=0$ and $\eta_2=1$.

\subsection{Hamiltonian structures of generalized D-KN integrable couplings}
\label{adkn}

Now, we begin with the enlarged spectral matrix of a generalized D-KN hierarchy (\ref{spic10}) and compute 
\begin{equation}
\langle \bar{W},\bar{U}_{\lambda} \rangle_{\tilde{\mathfrak{g}}(\lambda)}= (4a+4e)\lambda+fq_1+bq_2+cp_2+gp_1
\end{equation}
and
\begin{equation}
\langle \bar{W},\bar{U}_{\bar{u}} \rangle_{\tilde{\mathfrak{g}}(\lambda)}=[
g \lambda,
 f \lambda,
-2e,
g,
f ,
c \lambda ,
 b \lambda ,
-2a ,
c,
b ]^T.
\end{equation}
Substituting the Laurent series and comparing powers of $\lambda$, we have
\begin{equation}
\begin{split}
\frac{\delta}{\delta \bar{u}} \int \frac{(4a_{m+2}+4e_{m+2})+f_{m+1}q_1+b_{m+1}q_2+c_{m+1}p_2+g_{m+1}p_1}{m}dx= \\
[
g_{m+1},
 f_{m+1},
-2e_m,
g_m,
f_m ,
c_{m+1} ,
 b_{m+1} ,
-2a_m ,
c_m,
b_m ]^T, \enspace m \geq 1.
\end{split}
\end{equation}

A long calculation involving the recursion relations (\ref{rr4}) shows that 
\begin{equation}
\label{recstrut2}
\frac{\delta \mathcal{\bar{H}}_{m+1}}{\delta \bar{u}}=\bar{\Psi} \frac{\delta \mathcal{\bar{H}}_m}{\delta \bar{u}},
\end{equation}
where
\begin{equation}
\bar{\Psi}=\bar{\Phi}^{\dagger}=\begin{bmatrix} \Phi^\dagger & (\Phi_1-\Phi)^\dagger \\ 0 & \Phi^\dagger \end{bmatrix},
\end{equation}
with $\Phi$ and $\Phi_1$ from (\ref{phi1}), (\ref{phi10a}) and (\ref{phi10b}), respectively.
We consequently obtain Hamiltonian structures for the hierarchy  of integrable couplings (\ref{shier1}), i.e.,
\begin{equation}
\label{hh}
\bar{u}_{t_m}=\bar{J} \frac{\delta \bar{\mathcal{H}}_m}{\delta \bar{u}}, \enspace m \geq 0,
\end{equation}
with the Hamiltonian functionals,
\begin{equation}
\bar{\mathcal{H}}_m= \int  \frac{(4a_{m+2}+4e_{m+2})+f_{m+1}q_1+b_{m+1}q_2+c_{m+1}p_2+g_{m+1}p_1}{m}dx, 
\end{equation}
for $m \geq 1$, and
\begin{equation}
\bar{\mathcal{H}}_0= \int [(\beta-\alpha)p_1q_1+\alpha(p_1q_2+p_2q_1)-2\beta r_1 - 2 \alpha r_2]dx
\end{equation}
calculated directly from $[g_1, f_1, -2e_0, g_0, f_0, c_1, b_1, -2a_0, c_0, b_0]^T$. The Hamiltonian operator in (\ref{hh}) is the block matrix of the form:
\begin{equation}
\bar{J}=\begin{bmatrix} 0 & J_1 \\ J_1 & J_2 \end{bmatrix}, 
\end{equation}
where
\begin{equation}
J_1=\begin{bmatrix} 0 & 2 & 0 & 0 & 0\\ -2 & 0 & 0 & 0 & 0\\ 0 & 0 & \frac{1}{2} \partial & s_1 & -v_1 \\
0 & 0 & -s_1 & 0 & \partial+2r_1 \\  0 & 0 & v_1 & \partial-2r_1  & 0 \end{bmatrix}, J_2=\begin{bmatrix} 0 & 2 & 0 & 0 & 0 \\ -2 & 0 & 0 & 0 & 0 \\ 0 & 0 & 0 & s_2 & -v_2 \\ 0 & 0 & -s_2 & 0 & 2r_2 \\
 0 & 0 & v_2 & -2r_2  & 0 \end{bmatrix}. 
\end{equation}
As a direct result of the Hamiltonian structures (\ref{hh}), the recursion structure (\ref{recurr}) and (\ref{recstrut2}), and the property $\bar{J}\bar{\Psi}=\bar{\Psi}^\dagger \bar{J}$, the hierarchy (\ref{shier1}) has the following commutativity of flows:
\begin{equation}
\{ \bar{\mathcal{H}}_k,\bar{\mathcal{H}}_l \}_{\bar{J}}= \int \left ( \dfrac{\delta \bar{\mathcal{H}}_k}{\delta \bar{u}} \right)^T \bar{J} \dfrac{\delta \bar{\mathcal{H}}_l }{\delta \bar{u}} dx = 0.
\end{equation}
We also have the commutativity of symmetries for $\{ \bar{K}_n \}$, i.e.,
\begin{equation}
[\bar{K}_k,\bar{K}_l]=\bar{K}_k'(\bar{u})[\bar{K}_l]-\bar{K}_l'(\bar{u})[\bar{K}_k]=0, \quad k,l \geq 0.
\end{equation}
Therefore, the hierarchy (\ref{shier1}) is Liouville integrable, as expected.

\subsection{Bi-Hamiltonian structures of the reduced integrable couplings}
\label{ardkn}

Next, we focus on the reduced spectral matrix (\ref{spic20}) and compute 
\begin{equation}
\langle \bar{W},\bar{U}_{\lambda} \rangle_{\tilde{\mathfrak{g}}(\lambda)}= (4a+4e)\lambda+fq_1+bq_2+cp_2+gp_1
\end{equation}
and 
\begin{equation}
\begin{split}
\langle \bar{W},\bar{U}_{\bar{u}} \rangle_{\tilde{\mathfrak{g}}(\lambda)}= & [
(a-e)q_1-aq_2+g \lambda,
(a-e)p_1-ap_2+ f \lambda,
g,
f , \\
&-aq_1+c \lambda ,
-ap_1+ b \lambda ,
c,
b ]^T.
\end{split}
\end{equation}
Again, we compare powers of $\lambda$ after substituting the Laurent series for $a, b, c, e, f, g$ to get
\begin{equation}
\begin{split}
\frac{\delta}{\delta \bar{u}} \int \frac{(4a_{m+2}+4e_{m+2})+f_{m+1}q_1+b_{m+1}q_2+c_{m+1}p_2+g_{m+1}p_1}{m}dx=\\
[(a_m-e_m)q_1-a_mq_2+g_{m+1},
(a_m-e_m)p_1-a_mp_2+ f_{m+1},
g_m,
f_m \\
-a_mq_1+c_{m+1},
-a_mp_1+ b_{m+1},
c_m, 
b_m ]^T, \enspace m \geq 1.
\end{split}
\end{equation}
Now using the recursion relations (\ref{rr8}), we have 
\begin{equation}
\label{recstrut4}
\frac{\delta \mathcal{\bar{H}}_{m+1}}{\delta \bar{u}}=\bar{\Psi}\frac{\delta \mathcal{\bar{H}}_{m}}{\delta \bar{u}},
\end{equation}
where
\begin{equation}
\bar{\Psi}=\bar{\Phi}^{\dagger}=\begin{bmatrix} \Phi^\dagger & (\Phi_1-\Phi)^\dagger \\ 0 & \Phi^\dagger \end{bmatrix},
\end{equation}
with $\Phi$ and $\Phi_1$ from (\ref{phi2}) and (\ref{phi20}), respectively.

We finally obtain the bi-Hamiltonian structure for the hierarchy of integrable couplings (\ref{shier2}),
\begin{equation}
\label{biham}
\bar{u}_{t_m}=\bar{J} \frac{\delta \bar{\mathcal{H}}_{m+1}}{\delta \bar{u}}=\bar{M} \frac{\delta \bar{\mathcal{H}}_{m}}{\delta \bar{u}}, \enspace m \geq 0,
\end{equation}
with the Hamiltonian functionals
\begin{equation}
\bar{\mathcal{H}}_m= \int  \frac{(4a_{m+2}+4e_{m+2})+f_{m+1}q_1+b_{m+1}q_2+c_{m+1}p_2+g_{m+1}p_1}{m}dx, 
\end{equation}
for $m \geq 1$, and the block Hamiltonian operators $\bar{J}$ given by
\begin{equation}
\label{barj}
\bar{J}=\begin{bmatrix} 0 & J \\ J & J \end{bmatrix}, J=\begin{bmatrix} 0 & 0 & 0 & 2\\ 0 & 0 & -2 & 0 \\ 0 & 2 & 0 &0 \\ -2 & 0 & 0 & 0\end{bmatrix},
\end{equation}
and $\bar{M}=\bar{\Phi}\bar{J}$ where $\bar{J}$ is above (\ref{barj}) and $\bar{\Phi}$ is the recursion operator (\ref{ro}) for the reduced integrable couplings (\ref{shier2}). Recall, a bi-Hamiltonian property means that $\bar{J}$ and $\bar{M}$ constitute a Hamiltonian pair, or, $\bar{N}=\alpha \bar{J} + \beta \bar{M}$, for any $\alpha, \beta \in \mathbb{R}$, is a Hamiltonian operator. As a direct result of the bi-Hamiltonian structure (\ref{biham}), we can say that the soliton hierarchy (\ref{shier2}) is integrable in the Liouville sense:
\begin{equation}
\begin{cases}
\{ \bar{\mathcal{H}}_k,\bar{\mathcal{H}}_l \}_{\bar{J}}=& \int \left ( \dfrac{\delta \bar{\mathcal{H}}_k}{\delta \bar{u}} \right)^T \bar{J} \dfrac{\delta \bar{\mathcal{H}}_l }{\delta \bar{u}} dx = 0,  \\

\{ \bar{\mathcal{H}}_k,\bar{\mathcal{H}}_l \}_{\bar{M}}=& \int \left ( \dfrac{\delta \bar{\mathcal{H}}_k}{\delta \bar{u}} \right)^T \bar{M} \dfrac{\delta \bar{\mathcal{H}}_l }{\delta \bar{u}} dx = 0,
\end{cases}
\end{equation}
and
\begin{equation}
[\bar{K}_k,\bar{K}_l]=\bar{K}_k'(\bar{u})[\bar{K}_l]-\bar{K}_l'(\bar{u})[\bar{K}_k]=0, \quad k,l \geq 0.
\end{equation}

 \section{Concluding remarks}

We have introduced a new spectral matrix that is a generalization of the D-Kaup-Newell and Kaup-Newell spectral problems. Integrable couplings were a result  of solving zero-curvature equations on the enlarged zero-curvature equation and its reduction. The Hamiltonian structures of the integrable couplings were constructed and presented. The reduced hierarchy of integrable couplings was found to be bi-Hamiltonian and both hierarchies are Liouville integrable. These hierarchies are both new and different. Many new systems of equations have been found to be integrable like the coupled mKdV \ref{cmkdv}.

This paper uses a relatively new idea of having $\frac{\partial U_1}{\partial \lambda} \neq 0$ in the enlarged spectral matrix $\bar{U}$ where most integrable couplings are generated through perturbations and $\frac{\partial U_1}{\partial \lambda} = 0$. Although the calculations are more difficult and tedious, many new applications may arise from integrable couplings starting from enlarged spectral matrices of this form. One application is the Darboux transformation method for the construction of solutions to integrable couplings \cite{darboux, dissertation}. This construction creates new integrable systems associated with non-semisimple Lie algebras and brings us new insightful thoughts to classify integrable systems from an algebraic point of view.

\section*{References}

\bibliographystyle{iopart-num}

%\bibliography{...}
\end{document}